\documentclass[final, nomarks]{dmtcs-episciences}

\usepackage[utf8]{inputenc}
\usepackage{subfigure}

\usepackage[round]{natbib}

\usepackage{latexsym,amsthm,amsmath,amssymb,url}
\usepackage{paralist}
\usepackage{enumerate}
\usepackage{doi}

\newtheorem{theorem}{Theorem}
\newtheorem{lemma}{Lemma}

\newtheorem{corollary}{Corollary}

\usepackage{hyperref}

\newcommand{\arxiv}[1]{arXiv:\href{http://arxiv.org/abs/#1}{#1}}

\hypersetup{
  pdftitle = {Hitting minors, subdivisions, and immersions in tournaments},
  pdfauthor = {Jean-Florent Raymond},
}

\newcommand{\itemref}[1]{\hyperref[#1]{(\ref*{#1})}}

\DeclareMathOperator{\ctw}{\mathbf{ctw}}
\DeclareMathOperator{\pw}{\mathbf{pw}}

\newcommand{\N}{\mathbb{N}}
\newcommand{\intv}[2]{\left \{#1,\dots, #2 \right \}}

\title[Hitting minors, subdivisions, and immersions in tournaments]{Hitting minors, subdivisions, and immersions in tournaments\thanks{This work has been done while the author was affiliated to LIRMM, Montpellier, France and Institute of Computer Science, University of Warsaw, Poland and was supported by the grant PRELUDIUM 2013/11/N/ST6/02706 of the Polish National Science Center (NCN).}}

\author{Jean-Florent Raymond}

\affiliation{Technische Universit\"at Berlin, Germany}

\received{2017-2-13}

\revised{2017-12-20}

\accepted{2018-1-11}

\publicationdetails{20}{2018}{1}{5}{3139}

\begin{document}

\maketitle

\begin{abstract}
\noindent The Erdős--Pósa property relates parameters of covering and
packing of combinatorial structures and has been mostly studied in the
setting of undirected graphs. In this note, we use results of
Chudnovsky, Fradkin, Kim, and Seymour to show that, for every directed
graph $H$ (resp.\ strongly-connected directed graph $H$), the class of
directed graphs that contain $H$ as a strong minor (resp.\ butterfly minor,
topological minor)
has the vertex-Erdős--Pósa property in the class of tournaments. We
also prove that if $H$ is a strongly-connected directed graph, the
class of directed graphs containing $H$ as an immersion has the edge-Erdős--Pósa property in
the class of tournaments.
\end{abstract}

\section{Introduction}

We are concerned in this note with the Erdős--Pósa property in the setting of directed
graphs. This property, which has
mostly been studied on undirected graphs, is originated from the
following classic result by \cite{erdHos1965independent}:
there is a function $f \colon \N \to \N$,
such that, for every (undirected) graph $G$ and every positive integer $k$, one of
the following holds:
\begin{inparaenum}[(a)]
\item $G$ contains $k$ vertex-disjoint cycles; or
\item there is a set $X \subseteq V(G)$ with $|X| \leq f(k)$ and such that $G \setminus X$
  has no cycle.
\end{inparaenum}
This theorem expresses a duality between a parameter of \textit{packing}, the
maximum number of vertex-disjoint cycles in a graph, and a parameter
of \textit{covering}, the minimum number of vertices that intersect
all cycles.
This initiated a research line aimed at providing
conditions for this property to hold, for various combinatorial objects.
%
Formally, we say that a class of graphs $\mathcal{H}$ has the \emph{Erdős--Pósa
property} 
if there is a function $f \colon
\N \mapsto \N$ such that, for every positive integer $k$
and every graph $G$ (referred to as the \emph{host graph}) 
 one of the following holds:
\begin{itemize}
\item $G$ has $k$ vertex-disjoint subgraphs that are isomorphic to
  members of $\mathcal{H}$; or
\item there is a set $X \subseteq V(G)$ with $|X| \leq f(k)$ and such
  that $G \setminus X$ has no subgraph isomorphic to a member of~$\mathcal{H}$.
\end{itemize}
The Erdős--Pósa Theorem states that the class of cycles has this
property. 
One of the most general extensions of the Erdős--Pósa Theorem is
certainly the following byproduct of the Graph Minors series:
\cite{RobertsonS86GMV} proved that the class of graphs that contain $H$
as a minor have the Erdős--Pósa property iff $H$ is planar.

On the other hand, some classes like odd
cycles fail to have the Erdős--Pósa property, as proved by \cite{dejter1988unboundedness}.
When this happens, one can consider
particular classes of host graphs. In this direction,
\cite{reed1999mangoes} proved that odd cycles have the
Erdős--Pósa property in planar graphs.

A natural variant of the Erdős--Pósa property is to change, in the definition,
vertex-disjoint subgraphs for edge-disjoint ones and sets of vertices for
sets of edges. It has been proved that the Erdős--Pósa Theorem also holds
in this setting (see \cite[Exercise 5 of Section~9]{Diestel05grap}). Other results have been obtained about
this variant, less than on the vertex variant, though.
At this point we have to stress that, if the vertex and edge variants
of the Erdős--Pósa property have close definitions, one cannot in
general deduce one from the other. We refer the reader to the surveys of
\cite{Reed97tree} and \cite{Raymond2016recent} for more details about the
Erdős--Pósa property.

In the setting of directed graphs however, few results are
known. Until recently, the largest class of directed graphs that has been studied under
the prism of the Erdős--Pósa property was the class of directed cycles, see
\cite{ReedRST96pack, Reed1995gallai, Guenin2010packing,
  Seymour1996Packing, havet:hal-00816135}.
It is worth noting that,
besides its combinatorial interest, the Erdős--Pósa property in
directed graphs found applications in bioinformatics and in the study of
Boolean networks~\cite{Aracena2016RSnumber, Aracena1263580}.
We consider here finite directed graphs (\emph{digraphs}) that may have
multiple arcs, but not loops and we respectively denote by $V(G)$ and $E(G)$ the
set of vertices and the multiset of arcs of a digraph~$G$. A digraph
$G$ is said to be
\emph{strongly-connected} if it has at least one vertex and, for every $u,v
\in V(G)$, there is a directed path from $u$ to~$v$. In particular the
digraph with one vertex is strongly-connected. The most general
result about the Erdős--Pósa property in digraphs is certainly
the following directed counterpart of the aforementioned results of
Robertson and Seymour.

\begin{theorem}[\cite{AkhoondianKKW09thee}]\label{kreutzer}
  Let $H$ be a strongly-connected digraph that is a
  butterfly minor (resp.\ topological minor) of a cylindrical
  grid.\footnote{The notions of butterfly minor and topological minor
    will be defined in a forthcoming paragraph. We refer the reader to \cite{AkhoondianKKW09thee}
  for a definition of the cylindrical directed grid.} There is a function $f \colon \N \to \N$,
  such that for every digraph $G$ and every positive integer $k$, one of
  the following holds:
  \begin{itemize}
  \item $G$ has $k$ vertex-disjoint subdigraphs, each having $H$ as
    a butterfly minor (resp.\ topological minor); or
  \item there is a set $X \subseteq V(G)$ with $|X| \leq f(k)$ such that $G \setminus X$
    does not have $H$ as a butterfly minor (resp.\ topological minor).
\end{itemize}
\end{theorem}

On the other hand, \cite{AkhoondianKKW09thee} proved
that \autoref{kreutzer} does not hold for the strongly-connected
digraphs $H$ that do not satisfy the conditions of its statement. It seems therefore natural to ask under what restrictions
on the host digraphs the above result could be true for every
strongly-connected digraph, in the same spirit as the aforementioned
result of Reed.

The purpose of this note is twofold: obtaining new Erdős--Pósa type
results on directed graphs and providing evidence that techniques
analogues to those used in the undirected case may be adapted to the
directed setting. In particular, we describe conditions on the class
of host digraphs so that \autoref{kreutzer} holds for every
strongly-connected digraph~$H$.
Before we formally state our results, let us introduce some terminology.

Several directed counterparts of the notion of minor have been
introduced in the literature. An arc $(u,v)$ of a digraph is said to be
\emph{contractible} if either it is the only arc with head $v$, or it
is the only arc with tail $u$.
Following \cite{Johnson2001138} and \cite{Kim2015138}, we say that a
digraph $H$ is a \emph{butterfly minor} (resp.\ \emph{strong minor})
of a digraph $G$ if a digraph isomorphic to $H$ can be obtained from a
subdigraph of $G$ by contracting contractible arcs (resp.\
contracting strongly-connected subdigraphs to single vertices).
Notice that these notions are
incomparable. A motivation for these definitions is that taking
(butterfly  or strong) minors does not create directed cycles.
Unlike minors, immersions and topological minors are concepts
that are easily extended to the setting of directed graphs as they can be
defined in terms of paths.
We say that a digraph $H$ is a \emph{topological minor} of a digraph $G$ if there is a
subdigraph of $G$ that can be obtained from a digraph isomorphic to
$H$ by replacing arcs by directed paths (in the same direction) that do not share internal vertices. If we allow
these paths to share internal vertices but not arcs, then we say that
$H$ is an \emph{immersion} of~$G$. Observe that every topological minor is a butterfly minor.
However, as often with the Erdős--Pósa property, this does not allow
us in general to deduce an Erdős--Pósa-type result about the one
relation from a result about the other one.

Our results hold on superclasses of the extensively studied class of
\emph{tournaments}, that are all orientations of undirected complete
graphs.
For $s \in \N$, a $n$-vertex digraph is \emph{$s$-semicomplete} if
every vertex $v$ has at least $n-s$ (in- and out-) neighbors. A \emph{semicomplete} digraph is a 0-semicomplete
digraph. Note that a semicomplete digraph is not necessarily a tournament as it may
have multiple egdes between a pair of vertices. These classes generalize the class of tournaments.
Our contributions are the following two theorems.

\begin{theorem}\label{main:pw}
  For every digraph (resp.\ strongly-connected digraph) $H$ and every $s\in
  \N$, there is a function $f \colon \N \to \N$ such that for every
  $s$-semicomplete digraph $G$ and every positive integer $k$, one of
  the following holds:
  \begin{itemize}
  \item $G$ has $k$ vertex-disjoint subdigraphs, each containing $H$ as a strong minor (resp.\
    butterfly minor, topological minor); or\label{e:mainpw:first}
  \item there is a set $X \subseteq V(G)$ with $|X|\leq f(k)$ such
    that $G \setminus X$ does not contain $H$ as a
    strong minor (resp.\ butterfly minor, topological minor).
  \end{itemize}
\end{theorem}

\begin{theorem}\label{main:ctw}
  For every strongly-connected digraph $H$ on at
   least two vertices, there is a function $f \colon \N \to \N$ such
   that for every semicomplete digraph $G$ and every positive integer
   $k$, one of the following holds:
  \begin{itemize}
  \item $G$ has $k$ arc-disjoint subdigraphs, each containing $H$ as an immersion; or\label{e:maincw:first}
  \item there is a set $X \subseteq E(G)$ with $|X|\leq f(k)$ such
    that $G \setminus X$ does not contain $H$ as an immersion.
  \end{itemize}
\end{theorem}

\autoref{main:pw} and \autoref{main:ctw} can be easily extended to finite
families of graphs $H$, as noted in their proofs.
These theorems deal with the two variants of the Erdős--Pósa
property: the first one is related to vertex-disjoint
subdigraphs and sets of vertices (vertex version), whereas the second
one is concerned with arc-disjoint subdigraphs and sets of arcs (arc version).
In \autoref{main:ctw}, the requirement on the order of $H$ is necessary as we cannot
cover an arcless subdigraph (as the one-vertex digraph) with arcs.
%
Our proofs rely on exclusion results for the parameters of cutwidth
and pathwidth, that are stated in the sections where they are used.

The techniques that we use are originated from the undirected setting,
where they have been repeatedly applied (see for instance
\cite[(8.8)]{RobertsonS86GMV} and
\cite[Lemma~2.3]{FioriniJW12excl}).
They deal with
structural decompositions like tree decompositions or tree-cut
decompositions and their associated widths and can be informally
described as follows. If the
considered host graph has large width, then, using a structural result, we
can immediately conclude that it contains several disjoint subgraphs
of the desired type. Otherwise, the graph admits a structural
decomposition of small width, that can be used to find a small set of
vertices/edges covering all such subgraphs (see~\cite[Theorem
3.1]{Raymond2016recent} for an unified presentation in undirected graphs).
Similar ideas have been used in the context of
directed graphs in the proof of \autoref{kreutzer}. With this note, we
provide more examples of cases where the techniques used in the
undirected setting appear useful when dealing with digraphs.

\section{Hitting minors and subdivisions}

This section is devoted to the proof of \autoref{main:pw}.
For every $k\in \N$ and every graph $H$, we denote by $k \cdot H$ the disjoint union of $k$ copies of~$H$.
The structural decompositions that we use in this section are path
decompositions. Formally, a \emph{path-decomposition} of a digraph $G$
is a sequence $(X_1, \dots, X_r)$ of subsets of $V(G)$ satisfying the
following properties:
\begin{enumerate}[(i)]
\item $V(G) = \bigcup_{i=1}^r X_i$
\item for every arc $(u,v) \in E(G)$, there are integers $i$ and $j$ with $1 \leq j \leq i \leq r$
and such that $u \in X_i$ and $v
  \in X_j$;\label{e2}
\item for every $i,j \in \intv{1}{r}$, if a vertex $u \in V(G)$
  belongs to $X_i$ and $X_j$, then it also belongs to $X_k$ for every
  $k \in \intv{i}{j}$.
\end{enumerate}

The sets $\{X_i\}_{i\in \intv{1}{r}}$ are called \emph{bags} of the
path-decomposition. Intuitively, item \itemref{e2} asks that every arc of
$G$ either have its endpoints in some bag, or is oriented
``backwards''.
The \emph{width} of the above path-decomposition is defined as $\max_{i \in \intv{1}{r}} |X_i|-1$. The \emph{pathwidth} of $G$ is the minimum width over all path-decompositions of $G$.
The following properties of pathwidth are crucial in our proofs.
\begin{theorem}[\protect{\cite[Theorem 2.2.7]{kim2013containment},
  \cite[(1.1)]{Fradkin2013tourn}, \cite[(1.4)]{Kim2015138}}]\label{omega}
  For every digraph $H$, there is a positive integer $w$ such
  that every semicomplete $G$ that has pathwidth more than $w$
  contains $H$ as a strong minor, butterfly minor and topological minor.
\end{theorem}

\begin{theorem}[{\cite[Theorem~2]{Kitsunai2015}}]\label{th:kit}
For every $s,w \in \N$, there is a positive integer $w'$ such that
every $s$-semicomplete digraph with pathwidth at least $w'$ has a
subdigraph that is semicomplete and is of pathwidth at least $w$.
\end{theorem}

\begin{corollary}\label{c:zeta}
  For every $s\in \N$ and every digraph $H$, there is a constant $\zeta_{s,H}$
  such that every $s$-semicomplete digraph that has pathwidth at least
  $\zeta_{s,H}$ contains $H$ as a strong minor, butterfly minor and
  topological minor.
\end{corollary}

A classic result states that if a collection of subpaths of a path
does not contain more than $k$ vertex-disjoint elements, then there is
a set of $k$ vertices meeting all the
subpaths (see~\cite{gyarfas1970helly}).
We use here the following generalization of the above statement, due to Alon.

\begin{lemma}[\cite{alon98piercing}]\label{l:alonpiercing}
  Let $P$ be a path (undirected) and let $\mathcal{P}$ be a collection of subgraphs
  of $P$ that does not contain $k+1$ pairwise vertex-disjoint members. Then
  there is a set of $2p^2k$ vertices of $P$ meeting every element of
  $\mathcal{P}$, where $p$ is the maximal number of connected
  components of a graph in~$\mathcal{P}$.
\end{lemma}

A \emph{strongly-connected component} of a digraph is a maximal
subdigraph that is strongly-connected. Observe that a single vertex
may be a strongly-connected component.
If a subdigraph of a digraph $G$ is isomorphic to some member of a digraph
class $\mathcal{H}$, we call it an \emph{$\mathcal{H}$-subdigraph} of~$G$.

\begin{lemma}\label{boundtw-minor}
Let $\mathcal{H}$ be a (possibly infinite) class of digraphs with at most $p$
strongly-connected components.
For every digraph $G$ and a positive integer $k$, one of the following holds:
\begin{inparaenum}[(a)]
\item $G$ contains $k$ pairwise vertex-disjoint $\mathcal{H}$-subdigraphs; or\label{e:pack}
\item there is a set $X \subseteq V(G)$ with $|X|\leq 2p^2(k-1)(\pw(G)+1)$\label{e:cover}
  such that $G\setminus X$ has no $\mathcal{H}$-subdigraph.
\end{inparaenum}
\end{lemma}

\begin{proof}
We proceed by induction on $k\in \N$.
The base case $k=1$ is trivial. Let us prove the statement of the
lemma for
$k>1$ assuming that it holds for all lower values of~$k$ (induction
step). For this we consider a digraph $G$ such that \itemref{e:pack} does not hold.
Let $(X_1, \dots, X_l)$ be a path decomposition of~$G$ of minimum
width. Let $P$ be the undirected path on vertices $v_1, \dots, v_l$,
in this order.
For every subdigraph $F$ of $G$, we set:
\begin{align*}
A_F &= \{i \in \intv{1}{l},\ V(F) \cap X_i
  \neq \emptyset\}\quad \text{and}\\
P_F &= P \left [ \{v_i,\ i \in A_F\}\right ].  
\end{align*}

In other words, $A_F$ is the set of indices of the bags met by $F$ and
$P_F$ is the subgraph of $P$ induced by the vertices with these
indices. 
For every $\mathcal{H}$-subdigraph $H$ of $G$, we consider the
subgraph $P_H$ of~$P$. Let us denote by $\mathcal{P}$ the class of all such
graphs (which is finite). Notice that for every pair $F,F'$ of subdigraphs of $G$, if
$P_F$ and $P_{F'}$ are vertex-disjoint, then so are $F$ and~$F'$.
Using our initial assumption on $G$, we deduce that
$\mathcal{P}$ does not contain $k$ pairwise vertex-disjoint members.
Besides, if $F$ is strongly-connected, then $P_F$ is
connected. Moreover, if $F$ has at most $p$ strongly-connected
components, then $P_F$ has at most $p$ connected components.
Hence, every member of $\mathcal{P}$ has at most $p$ connected components.

By the virtue of \autoref{l:alonpiercing}, there is a set $Q$ of
$2p^2(k-1)$ vertices of $P$ such that $P \setminus Q$ does not contain
a subgraph of~$\mathcal{P}$. Let $X = \bigcup_{i\in \{j \in
  \intv{1}{l},\ v_j \in Q\}} X_i$. Let us show that $X$ satisfies the
requirements of \itemref{e:cover}. By contradiction, we assume that $G
\setminus X$ has an $\mathcal{H}$-subdigraph $H$. Then $P_H \in
\mathcal{P}$. Let $v_i$ be a vertex of $V(P_H) \cap Q$, which, by
definition of $Q$, is not empty. Then $X_i \subseteq X$ and $V(H) \cap
X_i \neq \emptyset$. This contradicts the fact that $H$ is a
subdigraph of $G \setminus X$. Consequently, $X$ is as required. As it
is the union of $2p^2(k-1)$ bags of an optimal path decomposition of
$G$, we have $|X| \leq 2p^2(k-1)(\pw(G)+1)$. This concludes the proof. 
\end{proof}

We would like to mention that a weaker form of \autoref{boundtw-minor} where
$\mathcal{H}$ consists of digraphs whose connected components are
strongly-connected can be obtained by adapting the ideas used by
\cite[(8.8)]{RobertsonS86GMV}, with a dependency in $p$ that is linear
instead of quadratic.

\begin{lemma}\label{l:cc-sc}
  Let $G$ be a digraph and let $H$ be the digraph obtained by contracting
  one strongly-connected subdigraph $S$ of $G$ to one single vertex $v_S$. Then $H$ and $G$ have
  the same number of strongly-connected components.
\end{lemma}

\begin{proof}
  Let $f$ be the map such that, for every $C \subseteq V(H)$ that induces a
  strongly-connected component,
  \[f(C) = \left \{
      \begin{array}{ll}
        C&\text{if $v_S \not \in C$}\\
        (C \setminus \{v_S\}) \cup S&\text{otherwise}
      \end{array}
    \right .\]
  Let $C$ be a subset of $V(H)$ that induces a
  strongly-connected component and let us show that $f(C)$ induces a
  strongly-connected subdigraph of~$G$. For this, we show that, for any
  $x,y \in f(C)$, there is a directed path from $x$ to $y$. If $x,y
  \in S$, this is true since in this case, $S \subseteq f(C)$ and $G[S]$ is
  strongly-connected. If none of $x,y$ belongs to $S$, they are both
  vertices of $C$ as well. Let $v_0 \dots v_l$ be a directed path from
  $x = v_0$ to $y=v_l$ in $H$. Notice that $v_0, \dots, v_l \in C$. If $v_S$ does not belong to this path,
  then this is a path of $G$ as well and we are done. Otherwise, let
  $i$ be such that $v_S = v_i$. By definition of $H$, there are arcs
  $(v_{i-1},u)$ and $(u', v_{i+1})$ in $G$, for some $u,u' \in S$. As
  $G[S]$ is strongly-connected, it contains a directed path $Q$ from
  $u$ to $u'$. Therefore, concatenating $v_0\dots v_{i-1}u$, $Q$, and
  $u'v_{i+1}\dots v_l$ yields a path from $x$ to $y$. The case where
  exactly one of $x,y$ belongs to $S$ is similar. Consequently, $f(C)$ induces a
  strongly-connected subdigraph in~$G$.

  Let us now show that $f(C)$ is
  a strongly-connected component. By contradiction, let us assume that
  there is in $G$ a directed walk $u_0 \dots u_l$ with $l >1$ such that $\{u_0 \dots u_l\}
  \cap f(C) = \{u_0, u_l\}$. If $S \cap (f(C) \cup \{u_0 \dots
  u_l\}) = \emptyset$ then $C$ does not induce a strongly-connected component
  of $H$, a contradiction. If $S \cap f(C) = \emptyset$, then $f(C) =
  C$ and then $G$ has
  a path (that is a subpath of $u_0 \dots u_l$) from a vertex of $f(C)$
  to one of $S$ and vice-versa. Therefore, there is in $H$ a path from a
  vertex of $C$ to $v_S$ and vice-versa, which contradicts the
  definition of~$C$. Thus $S$ intersects $f(C)$: by
  definition of $f$ we have $S \subseteq f(C)$ and $v_S \in C$. We
  deduce that $P=u_0 \dots u_l$ (or $P=v_S, u_1 \dots u_l$, resp.\ $P=u_0 \dots
  u_{l-1}, v_S$ if $u_0\in S$, resp.\ $u_l \in S$) is an oriented walk of $H$ on at least 3
  vertices that has its endpoints in $C$ and contains vertices that do not belong to $C$.
  This is not possible since $C$
  induces a strongly-connected component of~$H$. We deduce that $f(C)$
  is a strongly-connected component of~$G$.
   
  The function $f$ is clearly
  injective. Let us show that it is surjective. Let now $C\subseteq V(G)$ be a
  strongly-connected component of $G$. If $C$ contains a vertex of
  $S$, then $S \subseteq C$ as $C$ is a maximal strongly-connected
  subdigraph and $S$ is
  strongly-connected. In this case observe that $f( C \setminus S) \cup
  \{v_S\} = C$. Otherwise, $C \cap S = \emptyset$ and $f(C) = C$.
  Strongly-connected components of $G$ are in bijection with those of
  $H$, hence they are equally many.
\end{proof}

\begin{corollary}\label{c:sconn}
  Let $H$ be a digraph and let $G$ be a subdigraph-minimal digraph
  containing $H$ as a strong minor. Then $H$ and $G$ have the same number of
  strongly-connected components.
\end{corollary}


\begin{lemma}\label{l:sm}
  For every (possibly infinite) family $\mathcal{H}$ of digraphs with bounded number of
  strongly-connected components and every $s\in
  \N$, there is a function $f \colon \N \to \N$ such that, for every
  $s$-semicomplete digraph $G$ and every positive integer $k$, one of
  the following holds:
  \begin{inparaenum}[(a)]
  \item $G$ contains $k$ vertex-disjoint subdigraphs, each having a
    digraph of $\mathcal{H}$ as a strong minor; or\label{e:first}
  \item there is a set $X \subseteq V(G)$ with $|X|\leq f(k)$ such
    that $G \setminus X$ contains no digraph of $\mathcal{H}$ as a
    strong minor.\label{e:second}
  \end{inparaenum}
\end{lemma}

\begin{proof}
We prove the lemma for $f(k) = 2p^2(k-1) \zeta_{s, k\cdot H}$, where $p$ denotes the maximum number of strongly-connected components of a digraph in $\mathcal{H}$. Let us assume that \itemref{e:first} does not hold (otherwise we are done).
  Let $H \in \mathcal{H}$.
  According to \autoref{c:zeta}, we have $\pw(G) < \zeta_{s, k\cdot
    H}$.
  Let
  $\hat{\mathcal{H}}$ be the class of all subdigraph-minimal digraphs
  containing a digraph of $\mathcal{H}$ as a strong minor. Observe that
  $G$ has a digraph of $\mathcal{H}$ as a strong minor iff it has a
  subgraph isomorphic to a digraph in $\hat{\mathcal{H}}$. Also,
  according to \autoref{c:sconn}, the digraphs in $\hat{\mathcal{H}}$
  have at most $p$ strongly-connected components.
  We can now apply \autoref{boundtw-minor} and obtain a set $X$ of at most
  $2p^2(k-1) \zeta_{s, k\cdot H}$ vertices such that $G \setminus X$ contains no digraph of $\mathcal{H}$ as a
    strong minor, that is, item \itemref{e:second}. This concludes the proof.
\end{proof}

In general, subdigraph-minimal digraphs containing a digraph $H$ as a butterfly minor (resp.\
topological minor) may have more strongly-connected components
than~$H$. Therefore we focus on strongly-connected digraphs where the
following result plays the role of \autoref{l:cc-sc}.
\begin{lemma}
  Let $G$ be a strongly-connected digraph and let $H$ be the digraph obtained by contracting
  a contractible arc $(s,t)$ of $G$. Then $H$ is strongly-connected.
\end{lemma}

\begin{proof}
  In the case where $H$ is a single vertex, it is strongly-connected and we are done. So we now assume that $H$ has at least two vertices.
  Towards a contradiction, let us assume that there are two vertices
  $x,y \in V(H)$ such that there is a directed path $v_1 \dots v_l$ from $x=v_1$ to $y=v_l$ in
  $G$ but not in $H$. As $G$ and $H$ differ only by the contraction of
  $(s,t)$, there are distinct $i,j \in \intv{1}{l}$ such that $s =
  v_i$ and $t = v_j$. If $i<j$, then $v_1 \dots v_{i-1}v_j\dots v_l$ is a
  directed path from $x$ to $y$ in $H$, a contradiction. Let us now
  assume that $i>j$. In order to handle the case where $i=l$ or $j=0$, we observe that since $G$ is strongly-connected,
  there are arcs $(v_0,v_1)$ and $(v_l, v_{l+1})$ (for some vertices
  $v_0,v_{l+1}$ that may belong to the path we consider).
  Now, $v_i$ is the tail of the two arcs $(v_i,v_j)$ and $(v_i,
  v_{i+1})$ and $v_j$ is the head of the two arcs $(v_i,v_j)$ and
  $(v_{j-1}, v_j)$, which contradicts the contractibility
  of~$(s,t)$. Therefore, $H$ is strongly-connected.
\end{proof}

\begin{corollary}\label{c:scbutterfly}
  Let $H$ be a digraph whose connected components are
  strongly-connected and let $G$ be a subdigraph-minimal digraph
  containing $H$ as a butterfly minor (resp.\ topological minor). Then $H$ and $G$ have the same number of
  strongly-connected components.
\end{corollary}

\begin{lemma}\label{l:bm}
  For every finite family $\mathcal{H}$ of digraphs whose connected
  components are strongly-connected and every $s\in
  \N$, there is a function $f \colon \N \to \N$ such that, for every
  $s$-semicomplete digraph $G$ and every positive integer $k$, one of
  the following holds:
  \begin{inparaenum}[(a)]
  \item $G$ contains $k$ vertex-disjoint subdigraphs, each having a
    digraph of $\mathcal{H}$ as a butterfly minor (resp.\ topological minor); or\label{e2:first}
  \item there is a set $X \subseteq V(G)$ with $|X|\leq f(k)$ such
    that $G \setminus X$ contains no digraph of $\mathcal{H}$ as a
    butterfly minor (resp.\ topological minor).\label{e2:second}
  \end{inparaenum}
\end{lemma}

\begin{proof}
We prove the lemma for $f(k) = 2p^2(k-1) \zeta_{s, k\cdot H}$.
  This proof is similar to that of \autoref{l:sm}.
  Again, we can assume that \itemref{e2:first} does not hold and
  deduce $\pw(G) < \zeta_{s, k\cdot
    H}$ from \autoref{c:zeta}, for some $H \in \mathcal{H}$. We denote by $p$ the maximum  number of
  connected components of a digraph in $\mathcal{H}$ and by
  $\hat{\mathcal{H}}$
 the class of all subdigraph-minimal digraphs
  containing a digraph of $\mathcal{H}$ as a butterfly minor (resp.\
  topological minor). The digraphs in $\hat{\mathcal{H}}$
  have at most $p$ strongly-connected components, according to \autoref{c:scbutterfly}.
  We now apply \autoref{boundtw-minor} and obtain a set $X$ of at most
  $2p^2(k-1) \zeta_{s, k\cdot H}$ vertices satisfying item~\itemref{e2:second}.
\end{proof}

The part of \autoref{main:pw} related to strong minors is a
consequence of \autoref{l:sm} and \autoref{c:sconn}.
The part related to butterfly minors and topological minors follows
from \autoref{l:bm}.

\section{Hitting immersions}

This section is devoted to the proof of \autoref{main:ctw}.
For every two
subsets $X,Y \subseteq V(G)$, we denote by $E_G(X,Y)$ the set of arcs
of $G$ of the form $(x,y)$ with $x \in X$ and $y \in Y$.
Recall that an \emph{$\mathcal{H}$-subdigraph} of a digraph $G$ is any subdigraph of
$G$ that is isomorphic to some digraph in~$\mathcal{H}$.
The parameter that plays a major role in this section is cutwidth.
If $G$ is a digraph on $n$ vertices, the \emph{width} of an
ordering $v_1 \dots, v_n$ of its vertices is defined as
\[
\max_{i \in \intv{2}{n}} \left |E_G(\{v_1, \dots, v_{i-1}\}
  , \{v_i, \dots, v_n\})\right |.
\]
The \emph{cutwidth} of $G$, that we write
$\ctw(G)$, is the minimum width over all orderings~$V(G)$.
Intuitively, a digraph that has small cutwidth has an ordering where the
number of "left-to-right" arcs is small.
The following result plays a similar role as \autoref{omega}
in the previous section.
\begin{theorem}[\protect{\cite[(1.2)]{Chudnovsky2012tourn}}]\label{th:chu}
  For every digraph $H$, there is a positive integer $\eta_H$ such
  that every semicomplete digraph $G$ that has cutwidth at least $\eta_H$
  contains $H$ as an immersion.
\end{theorem}

\begin{proof}[Proof of \autoref{main:ctw}.]
Let ${\mathcal{H}}$ be the class of all
  subdigraph-minimal digraphs containing $H$ as
  an immersion and observe that these digraphs are strongly-connected.
Again, $G$ has contains $H$ as an immersion iff it has an
${\mathcal{H}}$-subdigraph.

According to \autoref{th:chu}, we are done if $\ctw(G) \geq \eta_{k\cdot
  H}$. Therefore we now consider digraphs of cutwidth at most $\eta_{k\cdot
  H}$.

We will prove the statement on digraphs of cutwidth at most $t$ by induction
on~$k$ with $f\colon k \mapsto k\cdot t$.
The case $k=0$ is trivial, therefore we assume $k>0$ and that the
result holds for every~$k'<k$. We also assume that $G$ does not
contain $k$ arc-disjoint ${\mathcal{H}}$-subdigraphs, otherwise we
are done.
Let $v_1, \dots, v_n$ be an ordering of the vertices of $G$ of minimum width.
Let $i\in \N$ be the minimum integer such that
$G[v_1, \dots, v_i]$ has an ${\mathcal{H}}$-subdigraph, that we call~$J$.
Notice that $i>1$ as we assume that $H$ has at
least two vertices.
We set $Y = E_G(\{v_1, \dots, v_{i-1}\}, \{v_{i}, \dots, v_n\})$.
Notice that $|Y|\leq t$.
As
the digraphs in ${\mathcal{H}}$ are strongly-connected, any
${\mathcal{H}}$-subdigraph of $G \setminus Y$
 belongs to exactly one of $\{v_1, \dots, v_{i-1}\},
\{v_{i}, \dots, v_n\}$. By definition of $i$, every such subdigraph
belongs to $\{v_{i}, \dots, v_n\}$. Notice that every subdigraph of
$\{v_{i}, \dots, v_n\}$ is arc-disjoint with $J$.
Therefore $G[\{v_{i}, \dots, v_n\}]$ does not contains $k-1$
arc-disjoint ${\mathcal{H}}$-subdigraphs. It is clear that this
subdigraph has cutwidth at most~$t$. By induction hypothesis, there is
a set $Y' \subseteq E(G[\{v_{i+1}, \dots, v_n\}])$ such that
$G[\{v_{i+1}, \dots, v_n\}] \setminus Y'$ has no
${\mathcal{H}}$-subdigraph
and $|Y'| \leq (k-1)\cdot t$. We deduce that $G \setminus (Y \cup
Y')$ has no ${\mathcal{H}}$-subdigraph and that $|Y \cup Y'| \leq
k\cdot t$, as required. This concludes the induction.
We saw above that we only need to consider digraphs of cutwidth at
most $\eta_{k\cdot
  H}$ and we just proved that in this case there is a suitable set of
arcs of size at most~$k\cdot \eta_{k\cdot
  H}$. This concludes the proof. Observe that \autoref{main:ctw} also
holds when considering a finite family of strongly-connected graphs $\mathcal{F}$ instead of $H$. For this $\mathcal{H}$
should be defined as the subdigraph-minimal digraphs containing a
digraph of $\mathcal{F}$ as an immersion, and $H$ as any digraph of
$\mathcal{F}$. The proof then follows the exact same lines.
\end{proof}

\section{Discussion}
\label{disc}

In this note we obtained new Erdős--Pósa type
results about classes defined by the relations of strong minors,
butterfly minors, topological minors and immersions.
The restriction of the host class to tournaments (or slightly larger
classes) allowed us to obtain results for every strongly-connected
pattern $H$.
In particular, we provided conditions on the host class where \autoref{kreutzer} holds for every
strongly-connected digraph~$H$, which is not the case in general.
Our proofs support the claim that techniques analogue to those used in
the undirected case may be adapted to the 
directed setting.
Let us now highlight two directions for future research.

\subsection*{Optimization of the gap.}
The bounds on the function $f$ in our results (gap of the Erdős--Pósa
property) depend on the exclusion bounds of \autoref{omega} and
\autoref{th:chu}. Therefore, any improvement of these bounds yields an
improvement of~$f$.
The upper bound on $\eta_H$ of \autoref{th:chu} that can be obtained
from the proof of \cite{Chudnovsky2012tourn} is $72\cdot 2^{2h(h+2)} + 8\cdot 2^{h(h+2)}$, where $h = |V(H)| +
2|E(H)|$. As a consequence, we have $f(k) = 2^{O(k^2h^2)}$ in
\autoref{main:ctw}. It would be interesting to know whether a gap that
is polynomial in $k$ can be obtained. The same question can be asked
for \autoref{main:pw}, however the upper bound in \autoref{omega} that we can compute from the proof of \cite{Fradkin2013tourn} is large (triply exponential).

\subsection*{Generalization.}
The results presented in this note were related to (generalizations
of) semicomplete digraphs. One direction for future research would be
to extend them to wider classes of hosts. On the other hand, in
\autoref{main:pw}, we require the guest digraph to be strongly
connected when dealing with butterfly  and topological minors. It is natural to ask if we can drop this condition. This
would require a different proof as ours draws upon this condition.


\end{document}